\documentclass[11pt]{article}

\usepackage[margin=1in]{geometry}
\usepackage{amsmath,amssymb,amsthm}
\usepackage{booktabs}
\usepackage{graphicx}
\usepackage{hyperref}

\newtheorem{theorem}{Theorem}

\begin{document}

\title{R\'{e}nyi exponent landscape of multipartite entanglement \\in free-fermion systems}

\author{A.~Sokolovs}

\date{March 2026}

\maketitle

\begin{abstract}
We show that the R\'{e}nyi tripartite information $I_3^{(\alpha)}$ of free fermions exhibits a \emph{qualitatively} $\alpha$-dependent scaling at small Fermi momentum, in sharp contrast to bipartite entropy where only the prefactor changes.  In the rank-1 regime ($z = k_F w \ll 1$), $I_3^{(\alpha)}$ receives contributions from two competing channels---a fractional-moment channel $\sim z^\alpha$ (active for non-integer~$\alpha$) and a polynomial channel $\sim z^m$ from the first nonvanishing inclusion-exclusion moment $\sigma_m$---yielding the scaling exponent $\beta_m(\alpha) = \min(\alpha, m)$ for $m$-partite information of $m$ adjacent strips.  Integer R\'{e}nyi indices $\alpha = 2, 3, \ldots$ are anomalous: the fractional channel closes and the exponent jumps to $m$ or higher.  A direct consequence is a \emph{replica obstruction}: $I_m^{(n)}/I_m^{(1)} \sim z^{m-1} \to 0$ for all integer $n \geq 2$, so the leading von~Neumann signal cannot be reconstructed from integer R\'{e}nyi data at the level of leading scaling---a situation with no bipartite analog.  Conversely, negativity-based measures ($\alpha = 1/2$) give a $20\times$ enhanced signal compared to von~Neumann.  We derive the underlying product formula for the coefficient $c(w_A, w_B, w_D)$, prove an $m$-partite generating function for the inclusion-exclusion moments, and verify all results numerically to high precision.
\end{abstract}

\section{Introduction}

The tripartite information $I_3 = S_A + S_B + S_D - S_{AB} - S_{BD} - S_{AD} + S_{ABD}$ probes genuinely tripartite correlations~\cite{CalabreseCardy2004}.  For holographic states $I_3 \leq 0$ universally~\cite{HaydenHeadrickMaloney}; free fermions violate this~\cite{AgonBuenoCasini,TamClaassenKane}.

In a companion paper~\cite{paper1}, we showed that for three adjacent strips of width~$w$, $I_3 = \sum_{k_y} g(k_F(k_y)\, w)$ with $g(z) \approx cz$ at small $z = k_F w$ and $c = 3\ln(4/3)/\pi$.  A natural question: how does this scale for R\'{e}nyi-$\alpha$ entropy, and what does this imply for experiments and for the replica trick?

For \emph{bipartite} entanglement of generic gapless states, R\'{e}nyi index changes only the prefactor: $S^{(\alpha)} \sim f(\alpha) \cdot \ln L$, same scaling for all~$\alpha$.  R\'{e}nyi-dependent exponents have been found for bipartite entropy in specially engineered states~\cite{SuginoKorepin2018,BarghathiDelMaestro2025}: number-conserving constructions can give $S^{(1)} \sim \sqrt{\ell}\ln\ell$ while $S^{(2)} \sim \ln\ell$.  We show that for \emph{multipartite} information $I_m$ ($m \geq 3$) of \emph{generic} free fermions, the scaling exponent depends on~$\alpha$ continuously via an exact formula $\beta_m(\alpha) = \min(\alpha, m)$, driven by a fundamentally different mechanism: the $m{-}1$ algebraic cancellations in the inclusion-exclusion combination act as a filter whose effect depends on the analyticity class of the entropy function.  This has concrete consequences for experiments and for the replica trick.

\section{Setup}
\label{sec:setup}

Consider $m$ adjacent strips of widths $w_1, \ldots, w_m$ on a one-dimensional chain with Fermi momentum~$k_F$.  The $m$-partite information involves $2^m - 1$ entropies $S_X^{(\alpha)}$ with inclusion-exclusion coefficients: $I_m^{(\alpha)} = \sum_{X} (-1)^{|X|+1} S_X^{(\alpha)}$.  Each block $X$ has total width $n_X = \sum_{i \in X} w_i$.  The $p$-th inclusion-exclusion moment is $\sigma_p = \sum_X (-1)^{|X|+1} n_X^p$.

Two algebraic identities are fundamental:
\begin{equation}\label{eq:cancel}
\sigma_1 = \cdots = \sigma_{m-1} = 0, \qquad \sigma_m = (-1)^{m+1}\, m!\, \prod_{i=1}^m w_i \neq 0.
\end{equation}
These follow from the Euler finite-difference identity and hold for all widths.  The first nonvanishing moment $\sigma_m$ is $m$-linear (for $m = 3$: $\sigma_3 = 6\,w_A w_B w_D$).  The full hierarchy is captured by the generating function
\begin{equation}\label{eq:genfun}
\sum_{p=m}^{\infty} \frac{\sigma_p}{p!}\, t^p = (-1)^{m+1} \prod_{i=1}^{m} (e^{w_i t} - 1),
\end{equation}
which we verify to $10^{-17}$ relative accuracy for $m = 3, 4, 5$.

\section{Master asymptotic formula}
\label{sec:master}

At small $z = k_F w$, the sine-kernel correlation matrix of each block is dominated by a single eigenvalue $\lambda_0 = n k_F/\pi$ (rank-1 regime).  The R\'{e}nyi-$\alpha$ entropy function $h_\alpha(\lambda) = \ln[\lambda^\alpha + (1{-}\lambda)^\alpha]/(1{-}\alpha)$ behaves differently depending on whether $\alpha$ is integer:

\emph{Non-integer~$\alpha$:} $h_\alpha(\lambda)$ contains a genuine $\lambda^\alpha$ term (non-polynomial).  The $m{-}1$ cancellations~(\ref{eq:cancel}) kill all polynomial terms up to $\lambda^{m-1}$ but \emph{cannot kill} $\lambda^\alpha$ for non-integer~$\alpha$.

\emph{Integer $\alpha \geq 2$:} $\lambda^\alpha + (1{-}\lambda)^\alpha$ is a polynomial, so $h_\alpha$ is polynomial in~$\lambda$.  All terms below $\sigma_m$ are cancelled.

These combine into:
\begin{equation}\label{eq:master}
\boxed{I_m^{(\alpha)}(z) = \underbrace{A_\alpha \cdot z^\alpha}_{\text{fractional channel}} + \underbrace{B_m \cdot z^m}_{\text{polynomial channel}} + \cdots}
\end{equation}
where $A_\alpha \propto F_\alpha \equiv \sum_X (-1)^{|X|+1} n_X^\alpha$ is the \emph{fractional moment} (nonzero for non-integer~$\alpha$; e.g., $F_{1/2} = 3 - 3\sqrt{2} + \sqrt{3} = 0.489$ for equal strips), and $B_m \propto \sigma_m \neq 0$ is the polynomial-channel coefficient, present for all~$\alpha$.

\section{R\'{e}nyi exponent landscape}
\label{sec:renyi}

\begin{theorem}[Exponent landscape]\label{thm:renyi}
For $m$-partite information of $m$ adjacent strips:
\begin{equation}\label{eq:beta}
\beta_m(\alpha) = \min(\alpha,\, m) \qquad \text{for non-integer } \alpha > 0.
\end{equation}
At integer $\alpha = n \geq 2$, the fractional channel closes and $\beta_m(n) \geq m$, with anomalies: $\beta_3(2) = 3$, $\beta_3(3) = 4$; $\beta_4(2) = 4$, $\beta_4(3) = 4$, $\beta_4(4) = 5$.
\end{theorem}

\begin{proof}
From Eq.~(\ref{eq:master}): for non-integer~$\alpha < m$, the fractional channel dominates ($A_\alpha \neq 0$, $z^\alpha \gg z^m$), giving $\beta = \alpha$.  For non-integer $\alpha > m$, the polynomial channel dominates ($z^m \gg z^\alpha$); the coefficient $B_m \neq 0$ because $\sigma_m \neq 0$ (Eq.~\ref{eq:cancel}) and the $m$-th Taylor coefficient of $h_\alpha(\lambda)$ is nonzero for all non-integer~$\alpha$.  The crossover occurs at $\alpha = m$.

At integer $\alpha = n$: $F_n$ reduces to $\sigma_n$ (polynomial), so $A_n = 0$ for $n < m$ (by Eq.~\ref{eq:cancel}) and the fractional channel closes.  The exponent is then set by the polynomial expansion of~$h_n$.  For $n = 2$ ($m = 3$): $\sigma_2 = 0$ kills the $\lambda^2$ term, leaving $\sigma_3 \Rightarrow \beta = 3$.  For $n = 3$ ($m = 3$): expanding $\lambda^3 + (1{-}\lambda)^3 = 1 - 3\lambda + 3\lambda^2$ shows the cubic term vanishes identically (the polynomial has degree~2), so the $\sigma_3$ channel is closed and $\beta = 4$.  For $n \geq 4$: the cubic coefficient of $h_n$ is nonzero, giving $\beta = m$.
\end{proof}

We verify $\beta_m(\alpha)$ numerically for $m = 2, 3, 4, 5$ at 25+ values of~$\alpha$ spanning $0.3$--$5.0$, extracting~$\beta$ from $\ln|I_m^{(\alpha)}(k_{F,2})/I_m^{(\alpha)}(k_{F,1})|/\ln(k_{F,2}/k_{F,1})$ with adaptive $k_F$ ranges.

\begin{figure}[ht]
\centering
\includegraphics[width=\textwidth]{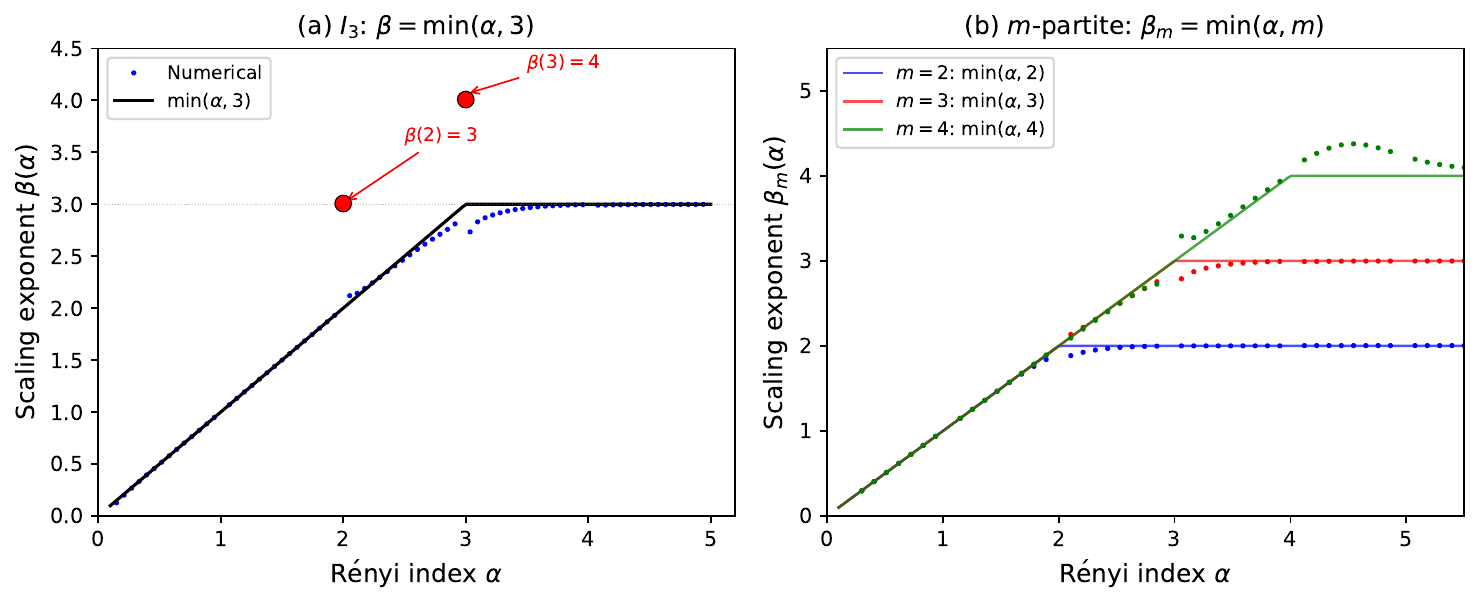}
\caption{\label{fig:renyi}(a)~$\beta(\alpha)$ for $I_3$ ($m = 3$).  Dots: numerical.  Line: $\min(\alpha, 3)$.  Red circles: integer anomalies at $\alpha = 2$ ($\beta = 3$) and $\alpha = 3$ ($\beta = 4$).  (b)~$m$-partite generalization: $\beta_m(\alpha) = \min(\alpha, m)$, crossover at $\alpha = m$.}
\end{figure}

\begin{center}
\begin{tabular}{c|cccccccc}
\toprule
 & \multicolumn{8}{c}{$\beta_m(\alpha)$, numerical} \\
$m$ & $\alpha{=}0.5$ & $1.0$ & $1.5$ & $2.0$ & $2.5$ & $3.0$ & $4.0$ & $5.0$ \\
\midrule
2 & 0.50 & 1.00 & 1.50 & 3.00$^\dag$ & 1.98 & 2.00 & 2.01 & 2.01 \\
3 & 0.50 & 1.00 & 1.50 & 3.01$^\dag$ & 2.48 & 4.01$^\dag$ & 2.99 & 2.99 \\
4 & 0.50 & 1.00 & 1.50 & 4.00$^\dag$ & 2.48 & 4.01$^\dag$ & 4.09$^\dag$ & 4.23 \\
5 & 0.50 & 1.00 & 1.50 & ---  & 2.49 & --- & 5.07$^\dag$ & 5.11 \\
\bottomrule
\end{tabular}
\end{center}

\noindent Daggers mark integer anomalies where $\beta$ exceeds $\min(\alpha, m)$.  Non-integer values agree with $\min(\alpha, m)$ to $1\%$.  The analysis is exact in the rank-1 limit $z \to 0$; at finite~$z$, corrections from higher Slepian eigenvalues enter at $O(z^{m+2}\ln z)$ but do not change the exponents.

\section{Consequences}
\label{sec:consequences}

\subsection{Replica obstruction}

The standard replica trick computes $\mathrm{Tr}(\rho^n)$ at integer $n = 2, 3, \ldots$ and continues to $n \to 1$.  For bipartite entropy, this works smoothly: $S^{(n)}/S^{(1)} \to O(1)$ for all~$n$.

For $I_m$ with $m \geq 3$:
\begin{equation}\label{eq:replica}
\frac{I_m^{(n)}(z)}{I_m^{(1)}(z)} \sim z^{\beta_m(n) - 1} \;\xrightarrow{z \to 0}\; 0 \qquad \text{for all integer } n \geq 2,
\end{equation}
since $\beta_m(n) \geq m \geq 3 > 1$.  The integer R\'{e}nyi data vanishes relative to the von~Neumann target.  Numerically, for $I_3$ at $z = 0.01$: $I_3^{(2)}/I_3^{(1)} = 10^{-4}$; for $I_4$: $I_4^{(2)}/I_4^{(1)} = 2 \times 10^{-6}$.  The obstruction grows as $z^{m-1}$.

The cause: integer $h_n$ is polynomial and ``sees'' only the polynomial channel ($\sigma_m$), while the von~Neumann $h_1 = -\lambda\ln\lambda$ accesses the singular channel that evades all $m{-}1$ cancellations.  Replica continuation must reconstruct a $z^1$ signal from $z^m$ data---from \emph{subleading corrections}, not from the leading term.

\subsection{Negativity enhancement}

For $\alpha < 1$, $\beta = \alpha < 1$: the signal is stronger than von~Neumann.  The case $\alpha = 1/2$ is physically relevant: R\'{e}nyi-$1/2$ entropy is related to the trace norm $\mathrm{Tr}\sqrt{\rho}$, which connects to entanglement negativity for pure-state bipartitions~\cite{VidalWerner2002,CalabreseCardyTonni2012}.  At $z = 0.01$:

\begin{center}
\begin{tabular}{lccc}
\toprule
Measure & $\alpha$ & $\beta$ & $|I_3^{(\alpha)}|$ \\
\midrule
Negativity-type & $1/2$ & $1/2$ & $5.5 \times 10^{-2}$ \\
Von Neumann & $1$ & $1$ & $2.7 \times 10^{-3}$ \\
R\'{e}nyi-2 & $2$ & $3$ & $2.7 \times 10^{-7}$ \\
\bottomrule
\end{tabular}
\end{center}

\noindent Negativity-based $I_3$ is $20\times$ stronger than von~Neumann and $2 \times 10^5$ stronger than R\'{e}nyi-2.  For detecting small Fermi pockets ($z \ll 1$), von~Neumann is the best \emph{standard} measure; negativity-type measures with $\alpha < 1$ would be optimal.

\subsection{R\'{e}nyi-2 blindness grows with $m$}

For integer R\'{e}nyi-2, $\beta_m(2) = m$ for $m \geq 3$: each additional party adds one power of~$z$ to the suppression.  R\'{e}nyi-2 is progressively blinder to higher multipartite correlations.  This is directly relevant for cold-atom experiments, where R\'{e}nyi-2 is the standard measurable via randomized protocols~\cite{Elben2023}.

\section{Product formula for the von Neumann coefficient}
\label{sec:product}

The von~Neumann coefficient $c$ that enters the fractional channel at $\alpha = 1$ has an exact closed form:

\begin{equation}\label{eq:product}
c(w_A, w_B, w_D) = \frac{1}{\pi}\ln\frac{(w_A{+}w_B)^{w_A+w_B}\,(w_B{+}w_D)^{w_B+w_D}\,(w_A{+}w_D)^{w_A+w_D}}{w_A^{w_A}\,w_B^{w_B}\,w_D^{w_D}\,S^S},
\end{equation}
where $S = w_A + w_B + w_D$.  This admits the additive decomposition $c = (S/\pi)\sum f(p_i)$ with $f(p) = (1{-}p)\ln(1{-}p) - p\ln p$ and $p_i = w_i/S$, maximized at the equal partition ($c_{\max} = S\ln(4/3)/\pi$).

A key ingredient is the \emph{disjoint-block degeneracy}: at rank-1, the separated block $A \cup D$ has the same eigenvalue as a contiguous block of the same width, because a constant kernel cannot distinguish separated from contiguous sites.  This produces the $S_3$ permutation symmetry (verified: all six permutations of $(1, 2, 3)$ give identical $g/k_F$ to machine precision).

The product formula yields a Widom-type factorization~\cite{GioevKlich2006} of the 2D tripartite information: $I_3 = (S\,H/4\pi^2)\, L\, A_F$, where $H(\{p_i\})$ depends only on partition fractions and $A_F$ is the Fermi sea area---not the boundary integral, which cancels due to $\sigma_1 = 0$.

\begin{table}[ht]
\caption{\label{tab:c}Product formula verified against Toeplitz numerics.  Accurate to $0.1\%$ at $k_F = 10^{-2}$, $4\%$ at $k_F = 10^{-1}$.  Error scales as $k_F^2|\!\ln k_F|$.}
\begin{center}
\begin{tabular}{ccc|cccc|c}
\toprule
$w_A$ & $w_B$ & $w_D$ & $k_F{=}10^{-3}$ & $10^{-2}$ & $5\!{\times}\!10^{-2}$ & $10^{-1}$ & $c$ \\
\midrule
1 & 1 & 1 & 0.2747 & 0.2748 & 0.2754 & 0.2767 & 0.2747 \\
2 & 3 & 5 & 0.8114 & 0.8124 & 0.8248 & 0.8454 & 0.8113 \\
3 & 3 & 3 & 0.8242 & 0.8251 & 0.8369 & 0.8570 & 0.8241 \\
\bottomrule
\end{tabular}
\end{center}
\end{table}

\section{Discussion}

The central result---$\beta_m(\alpha) = \min(\alpha, m)$---has a transparent origin in the master formula~(\ref{eq:master}).  The $m{-}1$ inclusion-exclusion cancellations act as a \emph{filter} that kills all polynomial contributions below~$\sigma_m$.  What passes through the filter depends on the analyticity of the entropy function: the singular $h_1 = -\lambda\ln\lambda$ and fractional-power $h_\alpha \sim \lambda^\alpha$ channels evade it; integer-$\alpha$ polynomial channels do not.

This creates a qualitative difference from bipartite entropy, where no such cancellations exist and $\beta(\alpha) = 1$ for all~$\alpha$ in generic gapless systems.  R\'{e}nyi-dependent exponents have been identified in bipartite entropy of specially engineered states~\cite{SuginoKorepin2018,BarghathiDelMaestro2025}, where the mechanism involves symmetry-resolved entanglement spectra with anomalous weight distributions.  The multipartite mechanism identified here is fundamentally different: it operates in \emph{generic} free-fermion ground states, requires no fine-tuning, and is controlled by the algebraic structure of the inclusion-exclusion combination rather than by the entanglement spectrum.  The difference grows with~$m$: for $I_m$, $m{-}1$ cancellations push the polynomial exponent to~$m$, widening the gap between integer and non-integer~$\alpha$.

The replica obstruction~(\ref{eq:replica}) has implications for field-theoretic calculations.  Standard replica-trick computations of multipartite information evaluate $I_m^{(n)}$ at integer~$n$ and continue to $n \to 1$.  Our result shows that the integer data scales as $z^m$ while the target scales as~$z$: the continuation must reconstruct the singular $-\lambda\ln\lambda$ channel from data that lacks it.  This is not merely a technical difficulty but a structural feature: the information needed for the $z^1$ coefficient resides in subleading corrections to $I_m^{(n)}$, not in its leading term.

\section{Conclusions}

We have established:

(i) A master asymptotic formula~(\ref{eq:master}) with two competing channels---fractional ($z^\alpha$) and polynomial ($z^m$)---yielding the R\'{e}nyi exponent $\beta_m(\alpha) = \min(\alpha, m)$ (Theorem~\ref{thm:renyi}).

(ii) A replica obstruction: $I_m^{(n)}/I_m^{(1)} \to 0$ for all integer $n \geq 2$, growing as $z^{m-1}$.

(iii) Negativity-based $I_3$ ($\alpha = 1/2$) is $20\times$ more sensitive than von~Neumann; R\'{e}nyi-2 is $10^5\times$ weaker.

(iv) An exact product formula for $c(w_A, w_B, w_D)$ with additive decomposition and $S_3$ symmetry from disjoint-block degeneracy.

(v) An $m$-partite generating function $(-1)^{m+1}\prod(e^{w_it}-1)$ for the inclusion-exclusion moments, with $m{-}1$ cancellations $\sigma_1 = \cdots = \sigma_{m-1} = 0$.

All results are verified numerically to high precision.



\begin{thebibliography}{10}

\bibitem{paper1}
A.~Sokolovs, arXiv:2603.03103 (2026).

\bibitem{CalabreseCardy2004}
P.~Calabrese and J.~Cardy,
J.\ Stat.\ Mech.\ \textbf{0406}, P06002 (2004).

\bibitem{HaydenHeadrickMaloney}
P.~Hayden, M.~Headrick, and A.~Maloney,
Phys.\ Rev.\ D \textbf{87}, 046003 (2013).

\bibitem{GioevKlich2006}
D.~Gioev and I.~Klich,
Phys.\ Rev.\ Lett.\ \textbf{96}, 100503 (2006).

\bibitem{VidalWerner2002}
G.~Vidal and R.~F.~Werner,
Phys.\ Rev.\ A \textbf{65}, 032314 (2002).

\bibitem{Elben2023}
A.~Elben, S.~T.~Flammia, H.-Y.~Huang, R.~Kueng, J.~Preskill, B.~Vermersch, and P.~Zoller,
Nat.\ Rev.\ Phys.\ \textbf{5}, 9 (2023).

\bibitem{SuginoKorepin2018}
F.~Sugino and V.~Korepin,
Int.\ J.\ Mod.\ Phys.\ B \textbf{32}, 1850306 (2018).

\bibitem{BarghathiDelMaestro2025}
H.~Barghathi and A.~Del Maestro,
arXiv:2512.24533 (2025).

\bibitem{AgonBuenoCasini}
C.~A.~Ag\'{o}n, P.~Bueno, and H.~Casini,
SciPost Phys.\ \textbf{12}, 153 (2022).

\bibitem{CalabreseCardyTonni2012}
P.~Calabrese, J.~Cardy, and E.~Tonni,
Phys.\ Rev.\ Lett.\ \textbf{109}, 130502 (2012).

\bibitem{TamClaassenKane}
P.~M.~Tam, M.~Claassen, and C.~L.~Kane,
Phys.\ Rev.\ X \textbf{12}, 031022 (2022).

\end{thebibliography}
\end{document}